\documentclass[preprintnumbers,amsmath,amssymb,superscriptaddress]{revtex4}
\usepackage{amssymb}
\usepackage{amsmath}
\usepackage{amsfonts}
\usepackage{graphicx}
\providecommand{\U}[1]{\protect\rule{.1in}{.1in}}
\newtheorem{theorem}{Theorem}

\newtheorem{lemma}[theorem]{Lemma}

\newtheorem{proposition}[theorem]{Proposition}
\newtheorem{remark}[theorem]{Remark}

\newenvironment{proof}[1][Proof]{\noindent\textbf{#1.} }{\ \rule{0.5em}{0.5em}}
\begin{document}

\title{On Symmetric SL-Invariant Polynomials in Four Qubits}   

\author{Gilad Gour}\email{gour@ucalgary.ca}
\affiliation{Institute for Quantum Information Science and 
Department of Mathematics and Statistics,
University of Calgary, 2500 University Drive NW,
Calgary, Alberta, Canada T2N 1N4} 
\author{Nolan R. Wallach}\email{nwallach@ucsd.edu}
\affiliation{Department of Mathematics, University of California/San Diego, 
        La Jolla, California 92093-0112}

\begin{abstract} 
We find the generating set of SL-invariant polynomials in four qubits that are also invariant under permutations of the qubits. The set consists of four polynomials of degrees 2,6,8, and 12, for which we find an elegant expression in the space of critical states. In addition, we show that the Hyperdeterminant in four qubits is the only SL-invariant polynomial 
(up to powers of itself) that is non-vanishing precisely on the set of generic states.
\end{abstract}

\maketitle

\section{Introduction}

With the emergence of quantum information science in recent years, much effort has been given to the study of entanglement~\cite{Hor09}; in particular, to its characterization, manipulation and quantification~\cite{Ple07}. It was realized that highly entangled states are the most desirable resources for many quantum information processing tasks.
While two-party entanglement was very well studied, entanglement in multi-qubit systems is far less understood.
Perhaps one of the reasons is that $n$ qubits (with $n>3$) can be entangled in an uncountable number of 
ways~\cite{Vid00,Ver03,GW} with respect to stochastic local operations assisted by classical communication (SLOCC). 
It is therefore not very clear what role 
entanglement measures can play in multi-qubits or multi-qudits systems unless they are defined operationally.
One exception from this conclusion are entanglement measures that are defined in terms of the absolute value of SL-invariant 
polynomials~\cite{GW,Ver03,Uh00,Luq03,Jens,W,Miy03,Cof00}.  

Two important examples are the 
concurrence~\cite{Woo98} and the square root of the 3-tangle (SRT)~\cite{Cof00}.
The concurrence and the SRT, respectively, are the only $\text{SL}(2,\mathbb{C})\otimes \text{SL}(2,\mathbb{C})$ and 
$\text{SL}(2,\mathbb{C})\otimes \text{SL}(2,\mathbb{C})\otimes\text{SL}(2,\mathbb{C})$ invariant measures of entanglement that are homogenous of degree 1. The reason is that in two and three qubit systems there is a unique homogeneous  SL-invariant polynomial (up to scalar multiple) such that all other homogeneous invariant polynomials are multiples of powers of it. However for 4 qubits or more, the picture is different since there are many algebraically independent homogenous SL invariant polynomials such as
the 4-tangle~\cite{Uh00} or the Hyperdeterminant~\cite{Miy03}. 

In this note, we find the generating set of all SL-invariant polynomials with the property that they are also invariant under any permutation of the four qubits. Such polynomials yield measure of entanglement that capture genuine 4 qubits entanglement. In addition, we show that the 4-qubit Hyperdeterminant~\cite{Miy03} is the only homogeneous SL-invariant polynomial (of degree 24) that is non-vanishing precisely on generic states. 

This note is written with a variety of audiences in mind. First and foremost are the researchers who study quantum entanglement. We have therefore endeavered to keep the mathematical prerequisates to a minimum and have opted for proofs that emphasize explicit formulas for the indicated $SL$-invariant polynomials. We are aware that there are shorter proofs of the main results using the improtant work of Vinberg. However, to us, the most important aspect of the paper is that the Weyl group of F4 is built into the study of entanglement for 4 qubits. Indeed, the well known result of Shepherd-Todd on the invariants for the Weyl group of F4 gives an almost immediate proof of Theorem 1. The referee has indicated a short proof of Theorem 5 using more algebraic geometry. Although our proof is longer, we have opted to keep it since it is more elementary. We should also point out that in the jargon of Lie theory the hyperdeterminant is just the discriminent for the symmetric space corresponing to SO(4;4).

\section{Symmetric Invariants}

Let $H_{n}=\otimes^{n}\mathbb{C}^{2}$ denote the space of $n$-qubits, and let $G=SL(2,\mathbb{C})^{\otimes n}$ act on $H_{n}$ by the tensor product action.
An $SL$-invariant polynomial, $f(\psi)$, is a polynomial in the components of the vector 
$\psi\in\mathcal{H}_n$, which is invariant under the action of the group $G$. That is, $f(g\psi)=f(\psi)$ for all $g\in G$. 
In the case of two qubits the SL invariant polynomials are polynomials in the $f_2$ given by the bilinear form $(\psi,\psi)$:
$$
f_2(\psi)\equiv (\psi,\psi)\equiv \langle\psi^{*}|\sigma_y\otimes\sigma_y|\psi\rangle\;\;,\;\;\psi\in\mathbb{C}^{2}\otimes\mathbb{C}^2\;,
$$
where $\sigma_y$ is the second $2\times 2$ Pauli matrix with $i$ and $-i$ on the off-diagonal terms.
Its absolute value is the celebrated concurrence~\cite{Woo98}. 

In the case of three qubits all SL invariant polynomials are polynomials in the $f_4$ described as follows:
$$
f_4(\psi)=\det\left[\begin{array}{cc} (\varphi_0,\varphi_0) & (\varphi_0,\varphi_1)\\ 
(\varphi_1,\varphi_0) &(\varphi_1,\varphi_1)\end{array}\right]\;,
$$
where the two qubits states $\varphi_i$ for $i=0,1$ are defined by the decomposition 
$|\psi\rangle=|0\rangle|\varphi_0\rangle+|1\rangle|\varphi_1\rangle$, and the bilinear form $(\varphi_i,\varphi_j)$ is defined above for two qubits.
The absolute value of $f_4$ is the
celebrated 3-tangle~\cite{CKW}. 

In four qubits, however, there are many independent $SL$-invariant polynomials and it is possible to show
that they are generated by four SL-invariant polynomials (see e.g.~\cite{GW} for more details and references).
Here we are interested in SL-invariant polynomials that are also invariant under the permutation of the qubits.

Consider the permutation group
$S_{n}$ acting by the interchange of the qubits. Let $\widetilde{G}$ be the
group $S_{n}\ltimes G.$ That is, the set $S_{n}\times G$ with multiplication
\[
(s,g_{1}\otimes g_{2}\otimes\cdots\otimes g_{n})(t,h_{1}\otimes h_{2}%
\otimes\cdots\otimes h_{n})=(st,g_{t1}h_{1}\otimes g_{t2}h_{2}\otimes
\cdots\otimes g_{tn}h_{n}).
\]
Then $\widetilde{G}$ acts on $H_{n}$ by these two actions. We are interested
in the polynomial invariants of this group action. 

One can easily check that $f_2$ and $f_4$ above are 
also $\widetilde{G}$-invariant. However, this automatic $\widetilde{G}$-invariance of $G$-invariants breaks down for
$n=4$. As is well known~\cite{GW}, the polynomials on $H_{4}$ that are invariant under $G
$ are generated by four polynomials of respective degrees $2,4,4,6$. For
$\widetilde{G}$ we have the following theorem:

\begin{theorem}
The $\widetilde{G}$-invariant polynomials on $H_{4}$ are generated by four
algebraically independent homogeneous polynomials $h_{1},h_{2},h_{3}$ and
$h_{4}$ of respective degrees $2,6,8$ and $12.$ Furthermore, 
the polynomials can be taken to be $\mathcal{F}_{1}(z),\mathcal{F}
_{3}(z),\mathcal{F}_{4}(z),\mathcal{F}_{6}(z)$ as given explicitly in Eq.~(\ref{fpoly}) of the proof.
\end{theorem}

\begin{proof}
To prove this result we will use some results from~\cite{GW}. Let
\[
u_{0}=\frac{1}{2}(\left\vert 0000\right\rangle +\left\vert 0011\right\rangle
+\left\vert 1100\right\rangle +\left\vert 1111\right\rangle ),
\]%
\[
u_{1}=\frac{1}{2}(\left\vert 0000\right\rangle -\left\vert 0011\right\rangle
-\left\vert 1100\right\rangle +\left\vert 1111\right\rangle ),
\]%
\[
u_{2}=\frac{1}{2}(\left\vert 0101\right\rangle +\left\vert 0110\right\rangle
+\left\vert 1001\right\rangle +\left\vert 1010\right\rangle ),
\]%
\[
u_{3}=\frac{1}{2}(\left\vert 0101\right\rangle -\left\vert 0110\right\rangle
-\left\vert 1001\right\rangle +\left\vert 1010\right\rangle ).
\]
Let $A$ be the vector subspace of $H_{4}$ generated by the $u_{j}$. Then $GA$
contains an open subset of $H_{4}$ and is dense. This implies that any
$G$-invariant polynomial on $H_{4}$ is determined by its restriction to $A$.
Writing a general state in $A$ as $z=\sum z_{i}u_{i}$, we can choose the invariant polynomials
such that their restrictions to $A$ are given by
\[
\mathcal{E}_{0}(z)=z_{0}z_{1}z_{2}z_{3},\;\;\;\mathcal{E}_{j}(z)=z_{0}^{2j}+z_{1}%
^{2j}+z_{2}^{2j}+z_{3}^{2j},\;\;j=1,2,3.
\]
In~\cite{GW} we give explicit formulas for their extensions to $H_{4}.$

Also, let $W$ be the group of transformations of $A$ given by $\{g\in
G|gA=A\}_{|A}$. \ Then $W$ is the finite group of linear transformations of
the form
\[
u_{i}\longmapsto\varepsilon_{i}u_{s^{-1}i}
\]
with $\varepsilon_{i}=\pm1$, $s\in S_{4}$ and $\varepsilon_{0}\varepsilon
_{1}\varepsilon_{2}\varepsilon_{3}=1$. One can show~\cite{W,GW} that every $W$-invariant
polynomial can be written as a polynomial in $\mathcal{E}_{0},\mathcal{E}
_{1},\mathcal{E}_{2},\mathcal{E}_{3}$. We now look at the restriction of the
$S_{4}$ that permutes the qubits to $A$. Set $\sigma_{i}=(i,i+1)$, $i=1,2,3$
where, say, $(2,3)$ corresponds to fixing the first and last qubit and
interchanging the second and third. Then they have matrices relative to the
basis $u_{j}:$

\[
\sigma_{1|A}=\left[
\begin{array}
[c]{cccc}%
1 & 0 & 0 & 0\\
0 & 1 & 0 & 0\\
0 & 0 & 1 & 0\\
0 & 0 & 0 & -1
\end{array}
\right]  ,
\]%
\[
\sigma_{2|A}=\frac{1}{2}\left[
\begin{array}
[c]{cccc}%
1 & 1 & 1 & 1\\
1 & 1 & -1 & -1\\
1 & -1 & -1 & 1\\
1 & -1 & 1 & -1
\end{array}
\right]  ,
\]%
\[
\sigma_{3|A}=\left[
\begin{array}
[c]{cccc}%
1 & 0 & 0 & 0\\
0 & 1 & 0 & 0\\
0 & 0 & 1 & 0\\
0 & 0 & 0 & -1
\end{array}
\right]  .
\]
Since $S_{4}$ is generated by $(1,2),(2,3),(3,4)$, it is enough to find those
$W$-invariants that are also $\sigma_{i|A}$ invariant for $i=1,2$. We note
that the only one of the $\mathcal{E}_{j}$ that is not invariant\ under
$\sigma_{1|A}$ is $\mathcal{E}_{0}$ and $\mathcal{E}_{0}(\sigma_{1}%
z)=-\mathcal{E}_{0}(z)$ for $z\in A$. Thus if $F(x_{0},x_{1},x_{2},x_{3})$ is
a polynomial in the indeterminates $x_{j}$ then $F(\mathcal{E}_{0}%
,\mathcal{E}_{1},\mathcal{E}_{2},\mathcal{E}_{3})$ is invariant under
$\nu=\sigma_{1|A} $ if and only if $x_{0}$ appears to even powers. It is an
easy exercise to show that if $\mathcal{E}_{4}=z_{0}^{8}+z_{1}^{8}+z_{2}%
^{8}+z_{3}^{8}$ then any polynomial in $\mathcal{E}_{0}^{2},\mathcal{E}%
_{1},\mathcal{E}_{2},\mathcal{E}_{3}$ is a polynomial in $\mathcal{E}%
_{1},\mathcal{E}_{2},\mathcal{E}_{3},\mathcal{E}_{4}$ and conversely (see the
argument in the very beginning of the next section). Thus we need only find
the polynomials in $\mathcal{E}_{1},\mathcal{E}_{2},\mathcal{E}_{3}%
,\mathcal{E}_{4}$ that are invariant under $\tau=\sigma_{2|A}$. A direct
calculation shows that $\mathcal{E}_{1}$ is invariant under $\tau$. Also,%

\[
\mathcal{E}_{1}^{2}\circ\tau=\mathcal{E}_{1}^{2}\;,\;\;\mathcal{E}_{2}\circ
\tau=\frac{3}{4}\mathcal{E}_{1}^{2}-\frac{1}{2}\mathcal{E}_{2}+6\mathcal{E}
_{0}\;,\;\;\mathcal{E}_{0}\circ\tau=-\frac{1}{16}\mathcal{E}_{1}^{2}+\frac{1}
{8}\mathcal{E}_{2}+\frac{1}{2}\mathcal{E}_{0}.
\]
Since $\mathcal{E}_{0},\mathcal{E}_{1}^{2},\mathcal{E}_{2}$ forms a basis of
the $W$ invariant polynomials of degree $4$ this calculation shows that the
space of polynomials of degree $4$ invariant under $\tau,\nu$ and $W$ (hence
under $\widetilde{W}$) consists of the multiples of $\mathcal{E}_{1}^{2}$. The
space of polynomials invariant under $W$ and $\nu$ and homogeneous of degree
$6$ is spanned by $\mathcal{E}_{1}^{3},\mathcal{E}_{1}\mathcal{E}
_{2}$, and $\mathcal{E}_{3}$. From this it is clear that the space of homogeneous degree 6 polynomials that are
invariant under $\widetilde{W}$ is two dimensional. Since $\mathcal{E}
_{1}^{3}$ is clearly $\widetilde{W}$-invariant there is one new invariant of
degree $6$. Continuing in this way we find that to degree 12 there are
invariants $h_{1},h_{2},h_{3},h_{4}$ of degrees $2,6,8$ and $12$, respectively,
such that: (1) the invariant of degree $8$, $h_{3}$, is not of the form $ah_{1}%
^{4}+bh_{1}h_{2}$, (2) there is no new invariant of degree $10$, and (3) the invariant of degree 12, $h_{4}$,
cannot be written in the form $ah_{1}^{6}+bh_{1}^{3}h_{2}+ch_{1}^{2}h_{3}$. To describe these invariants we write
out a new set of invariants. We put
\begin{equation}\label{fpoly}
\mathcal{F}_{k}(z)=\frac{1}{6}\sum_{i<j}\left(  z_{i}-z_{j}\right)
^{2k}+\frac{1}{6}\sum_{i<j}\left(  z_{i}+z_{j}\right)  ^{2k}.
\end{equation}
We note that $\mathcal{F}_{1}=\mathcal{E}_{1}$, $\mathcal{F}_{2}%
=\mathcal{E}_{1}^{2}$. A direct check shows that these polynomials are
invariant under $\widetilde{W}$. Since $\mathcal{F}_{3}(z)\neq c\mathcal{E}%
_{1}^{3}$ we can use it as the \textquotedblleft missing
polynomial\textquotedblright. If one calculates the Jacobian determinant of
$\mathcal{F}_{1}(z),\mathcal{F}_{3}(z),\mathcal{F}_{4}(z),\mathcal{F}_{6}(z)$
then it is not $0$. This implies that none of these polynomials can be
expressed as a polynomial in the others. Thus they can be taken to be
$h_{1},h_{2},h_{3},h_{4}$.

Let $A_{\mathbb{R}}$ denote the vector space over $\mathbb{R}$ spanned by the $u_{j}$. 
If $\lambda\in A_{\mathbb{R}}$ is non-zero then we set for $a\in A$, $s_{\lambda}a=a-\frac{2\left\langle
\lambda|a\right\rangle }{\left\langle \lambda|\lambda\right\rangle }\lambda$.
Then such a transformation is called a reflection. It is the reflection about
the hyperplane perpendicucular to $\lambda$. The obvious calculation shows
that $\nu=s_{u_{3}}$and if $\alpha=\frac{1}{2}(u_{0}-u_{1}-u_{2}-u_{3})$ then
$\tau=s_{\alpha}$. We note that $W$ is generated by the reflections
corresponding to $u_{0}-u_{1},u_{1}-u_{2},u_{2}-u_{3}$ and $u_{2}+u_{3}$. This
implies that the group $\widetilde{W}$ is generated by reflections. One also
checks that it is finite (actually of order $576$). The general theory (cf.~\cite{Hum}) implies that the algebra of
invariants is generated by algebraicly independent homogeneous polynomials.
Using this it is easy to see that $\mathcal{F}_{1}(z),\mathcal{F}%
_{3}(z),\mathcal{F}_{4}(z),\mathcal{F}_{6}(z)$ generate the algebra of
invariants.
\end{proof}

\begin{remark}
Alternatively, we note that $\widetilde{W}$ is isomorphic with the Weyl group
of the exceptional group $F_{4}$ (see Bourbaki, Chapitares 4,5, et 6 Planche
VIII pp. 272,273). The exponents (exposants on p.273) are 1,5,7,11. This
implies that the algebra of invariants is generated by algebraicly independent
homogeneous polynomials of degrees one more than the exponents so 2,6,8,12. We
also note that the basic invariants for $F_{4}$ were given as 
$\mathcal{F}_{1}(z),\mathcal{F}_{3}(z),\mathcal{F}_{4}(z),\mathcal{F}_{6}(z)$ for the
first time by M.L.Mehta, Comm. Algebra,16(1988), pp. 1083-1098.
\end{remark}

For $n\geq4$ qubits the analogue of the space $A$ would have to be of
dimension $2^{n}-3n.$ Thus even if there were a good candidate one would be
studying, say, for $5$ qubits, a space of dimension $17$ and an immense finite
group that cannot be generated by reflections.

\section{A special invariant (hyperdeterminat) for 4 qubits.}

In this section we show that the hyperdeterminant for qubits is the only polynomial that quantifies genuine 4-way generic entanglement.
We start by observing that Newton's formulas (relating power sums to
elementary symmetric functions) imply that if $a_{1},...,a_{n}$ are elements
of a an algebra over $\mathbb{Q}$ (the rational numbers) then
\[
a_{1}a_{2}\cdots a_{n}=f_{n}(p_{1}(a_{1},...,a_{n}),...,p_{n}(a_{1}
,...,a_{n}))
\]
with $f_{n}$ a polynomial with rational coefficients in $n$ indeterminates and
$p_{i}(x_{1},...,x_{n})=\sum x_{j}^{i}$. This says that in the notation of the
previous theorem
\[
\gamma(z)=
{\displaystyle\prod\limits_{i<j}}
(z_{i}-z_{j})^{2}(z_{i}+z_{j})^{2}
\]
is $\widetilde{W}$-invariant. Indeed, take $a_{1,...,}a_{n(n-1)}$ to be
$\{(z_{i}-z_{j})^{2}|i<j\}\cup\{(z_{i}+z_{j})^{2}|i<j\}$ in some order. We
will also use the notation $\gamma$ for the corresponding polynomial of degree
24 on $H_{4}$.

We define the generic set, $\Omega$, in $H_{4}$ to be the set of elements,
$v$, such that $\dim Gv$ is maximal (that is, $12$). Then every such element
can be conjugated to an element of $A$ by an element of $G$. It is easily
checked that
\[
\Omega\cap A=\{\sum z_{i}u_{i}|z_{i}\neq\pm z_{j}\text{ if }i\neq j\}.
\]
This implies that $\Omega=\{\phi\in H_{4}|\gamma(\phi)\neq0\}$.

\begin{proposition}
If $f$ is a polynomial on $H_{4}$that is invariant under the action of $G$ and
is such that $f(H_{4}-\Omega)=0$ then $f$ is divisible by $\gamma$.
\end{proposition}

\begin{proof}
Since $f(z)=0$ if $z_{i}=\pm z_{j}$ for $i\neq j$ we see that $f$ is divisible
by $z_{i}-z_{j}$ and $z_{i}+z_{j}$ for $i<j$. Thus if
\[
\Delta(z)={\displaystyle\prod\limits_{i<j}}
(z_{i}-z_{j})(z_{i}+z_{j})
\]
then $f=\Delta g$ with $g$ a polynomial on $A$. One checks that $\Delta
(sz)=\det(s)\Delta(z)$ for $s\in W$ (see the notation in the previous
section). Since $f(sz)=f(z)$ for $s\in W$ we see that $g(sz)=\det(s)g(z)$ for
$s\in W.$ But this implies that $g(z)=0$ if $z_{i}=\pm z_{j}$ for $i\neq j$.
So $g$ is also divisible by $\Delta$. We conclude that $f$ is divisible by
$\Delta^{2}$. This is the content of the theorem.
\end{proof}

\begin{lemma}
$\gamma$ is an irreducible polynomial.
\end{lemma}

\begin{proof}
Let $\gamma=\gamma_{1}\gamma_{2}\cdots\gamma_{m}$ be a factorization into
irreducible (non-constant) polynomials. If $g\in G$ then since the
factorization is unique up to order and scalar multiple there is for each
$g\in G$, a permutation $\sigma(g)\in S_{m}$ and $c_{i}(g)\in%
\mathbb{C}
-\{0\}$, $i=1,...,m$ such that $\gamma_{j}\circ g^{-1}=c_{j}(g)\gamma
_{\sigma(g)j}$ for $j=1,...,m$. The map $g\longmapsto\sigma(g)$ is a group
homomorphism. The kernel of $\sigma$ is a closed subgroup of $G$. Thus
$G/\ker\sigma$ is a finite group that is a continuous image of $G$. So it must
be the group with one element since $G$ is connected. \ This implies that each
$\gamma_{j}$ satisfies $\gamma_{j}\circ g^{-1}=c_{j}(g)\gamma_{j}$for all
$g\in G$. We therefore see that $c_{j}:G\rightarrow%
\mathbb{C}
-\{0\}$ is a group homomorphism for each $j$. But the commutator group of $G$
is $G$. Thus $c_{j}(g)=1$ for all $g$. This implies that each of the factors
$\gamma_{j}$ is invariant under $G$. Now each $\gamma_{j|A}$ divides
$\gamma_{|A}$ thus in must be a product
\[{\displaystyle\prod\limits_{i<j}}
(z_{i}-z_{j})^{a_{ij}}(z_{i}+z_{j})^{b_{ij}}
\]
by unique factorization. We note that if $i<j$ then $\{(z_{i}+z_{j})\circ
s|s\in W\}=\{(z_{i}-z_{j})\circ s|s\in W\}=\{\varepsilon(z_{i}+z_{j}%
)|i<j,\varepsilon\in\{\pm1\}\}\cup\{z_{i}-z_{j}|i\neq j\}$. This implies that
since $\gamma_{j|A}$ is non-constant and $W$-invariant that each $\gamma
_{j|A}$ must be divisble by $\Delta_{|A}$. Now arguing as in the previous
propositin the invariance implies that $\gamma_{j}$ is divisible by
$\Delta^{2}=\gamma$. This implies that $m=1$.
\end{proof}

\begin{theorem}
If $f$ is a polynomial on $H_{4}$ such that $f(\phi)\neq0$ for $\phi\in\Omega$
then there exists $c\in%
\mathbb{C}
$, $c\neq0$ and $r$ such that $f=c\gamma^{r}$.
\end{theorem}

\begin{proof}
We may assume that $f$ is non-constant. Let $h$ be an irreducible factor of
$f$. Then $h(\phi)\neq0$ if $\phi\in\Omega$. This implies that the irreducible
variety $Y=\{x\in H_{4}|h(x)=0\}\subset H_{4}-\Omega=\{x\in H_{4}%
|\gamma(x)=0\}.$ Since the variety $H_4 - \Omega$ is irreducible, and since both varieties are of dimension $15$ over $%
\mathbb{C}
$ they must be equal. This implies that $h$ must be a multiple of $\gamma$.
Since $f$ factors into irreducible non-constant factors the theorem follows.
\end{proof}

\section{Discussion}
In this note we have shown that the set of all 4-qubit SL-invariant polynomials that are also invariant under permutations
of the qubits is generated by four polynomials of degrees 2,6,8,12. Using a completely different approach, in~\cite{Jens} it was also shown that these polynomials exist but they were not given elegantly as in Eq.~(\ref{fpoly}). In addition, we have shown here that the Hyperdeterminant~\cite{Miy03} is the \emph{only} SL-invariant polynomial (up to its powers)
that is not vanishing precisely on the set of generic states.

Since the Hyperdeterminant (in our notations $\gamma(z)$) quantifies generic entanglement, a state with the most amount of generic entanglement can be defined as a state, $z$, that maximize $|\gamma(z)|$. We are willing to conjecture that the state
$$
|L\rangle=\frac{1}{\sqrt{3}}\left(u_0+\omega u_1+\bar{\omega}u_2\right)\;,\;\;\omega\equiv e^{i\pi/3}\;,
$$
is the unique state (up to a local unitary transformation) that maximizes $|\gamma(z)|$. Our conjecture is based on numerical tests. In addition, one can prove that $|L\rangle$ is a critical point for $|L\rangle$ in the sphere and the Hessian on the sphere is negative semi-definate.

It was shown in~\cite{GW} that the state $|L\rangle$ maximizes uniquely many measures of 4 qubits entanglement. Moreover, one can easily check that the state $|L\rangle$ is the only state for which $\mathcal{E}_{0}=\mathcal{E}_{1}=\mathcal{E}_{2}=0$ while 
$\mathcal{E}_{3}(|L\rangle)=1/9$. It is known that a state with such a property is unique and called cyclic~\cite{K1,K2}. Similarly, we found out the unique state, $|F\rangle$ for which $\mathcal{F}
_{1}(|F\rangle)=\mathcal{F}_{3}(|F\rangle)=\mathcal{F}_{4}(|F\rangle)=0$ but $\mathcal{F}_{6}(|F\rangle)\neq 0$. The (non-normilized) unique state (up to a local unitary transformation) is
$$
|F\rangle=(3-\sqrt{3})u_0+(1+i)\sqrt{3}u_1+(1-i)\sqrt{3}u_2-i(3-\sqrt{3})u_3\;.
$$

\emph{Acknowledgments:---} GG research is supported by NSERC.

\end{document}